\documentclass[journal]{IEEEtran}
\usepackage{graphicx}

\usepackage{bm}
\usepackage{amsfonts}
\usepackage{amsthm}
\usepackage{amssymb}

\usepackage{color}
\usepackage{mathrsfs}

\usepackage{cite}

\usepackage{booktabs}

\newtheorem{thm}{Theorem}
\newtheorem{lemma}{Lemma}

\ifCLASSINFOpdf
\else
\fi
\usepackage{amsmath}
\usepackage{amsthm}

\usepackage{algorithm}
\usepackage{algorithmic}
\ifCLASSOPTIONcompsoc
 \usepackage[caption=false,font=normalsize,labelfont=sf,textfont=sf]{subfig}
\else
 \usepackage[caption=false,font=footnotesize]{subfig}
\fi
\begin{document}
%
\title{Towards Dual-functional Radar-Communication Systems: Optimal Waveform Design}
%
%
%

\author{Fan Liu,~\IEEEmembership{Student~Member,~IEEE,}
        Longfei Zhou,~\IEEEmembership{Student~Member,~IEEE,}
        Christos Masouros,~\IEEEmembership{Senior~Member,~IEEE,}
        Ang Li,~\IEEEmembership{Student~Member,~IEEE,}
        Wu Luo,~\IEEEmembership{Member,~IEEE}
        and~Athina Petropulu,~\IEEEmembership{Fellow,~IEEE}
\thanks{F. Liu is with the School of Information and Electronics, Beijing Institute of Technology, Beijing, 100081, China, and is also with the Department of Electronic and Electrical Engineering, University College London, London, WC1E 7JE, UK (e-mail: liufan92@bit.edu.cn).}
\thanks{L. Zhou and W. Luo are with the State Key Laboratory of Advanced Optical Communication Systems and Networks, Peking University, Beijing 100871, China (e-mail: zhoulongfei@pku.edu.cn, luow@pku.edu.cn).}
\thanks{C. Masouros and A. Li are with the Department of Electronic and Electrical Engineering, University College London, London, WC1E 7JE, UK (e-mail: chris.masouros@ieee.org, ang.li.14@ucl.ac.uk).}
\thanks{A. P. Petropulu is with the Department of Electrical and Computer Engineering, Rutgers, the State University of New Jersey, 94 Brett Road, Piscataway, NJ 08854, United States (e-mail: athinap@rutgers.edu).}
\thanks{This work has been submitted to the IEEE for possible publication. Copyright may be transferred without notice, after which this version may no longer be accessible.}
}

\maketitle

\begin{abstract}
We focus on a dual-functional multi-input-multi-output (MIMO) radar-communication (RadCom) system, where a single transmitter communicates with downlink cellular users and detects radar targets simultaneously. Several design criteria are considered for minimizing the downlink multi-user interference. First, we consider both the omnidirectional and directional beampattern design problems, where the closed-form globally optimal solutions are obtained. Based on these waveforms, we further consider a weighted optimization to enable a flexible trade-off between radar and communications performance and introduce a low-complexity algorithm. The computational costs of the above three designs are shown to be similar to the conventional zero-forcing (ZF) precoding. Moreover, to address the more practical constant modulus waveform design problem, we propose a branch-and-bound algorithm that obtains a globally optimal solution and derive its worst-case complexity as a function of the maximum iteration number. Finally, we assess the effectiveness of the proposed waveform design approaches by numerical results.
\end{abstract}


%
\IEEEpeerreviewmaketitle

\section{Introduction}
%
%
%
%
\IEEEPARstart{I}{t} has been reported that by 2020, the number of connected devices will jump to more than 20 million, which brings forward impending needs for extra frequency spectrum resources. Realizing the scarcity of the spectrum, network providers and policy regulators are exploring the feasibility to use the spectrum that is currently occupied by other applications \cite{federal2010connecting,ears,ssparc,specees}, such as airborne radars and navigation systems close to the 3.4GHz band\cite{ofcom2015} and shipborne and Vessel Traffic Service (VTS) radar at 5.6GHz\cite{euro2016}, which may be shared with LTE and Wi-Fi systems in the near future. As an emerging research topic, the communications-radar spectrum sharing (CRSS) not only presents the advantage for enabling the efficient usage of the spectrum, but also provides a new way for designing novel systems that can benefit from the cooperation of radar and communications.
\\\indent As a straightforward way to achieve the spectral coexistence for communication and radar, opportunistic spectrum sharing \cite{6331681} provides a naive approach, where the communication system transmits when the space and frequency spectra are not occupied by radar. Nevertheless, it does not allow the two systems to work simultaneously. In view of this, the pioneering work \cite{6503914} proposes a null-space projection (NSP) method, which has been widely applied to different spectral coexistence scenarios between MIMO radar and communication systems\cite{7089157,7814210}. In such schemes, a radar beamformer is designed to project the signals onto the null-space of the interference channel between the radar and base station (BS)/user equipment (UE), such that the interference from the radar to the communication link is zero. This, however, results in performance loss for the radar, since the beamforming is no longer optimal for target detection and estimation. Trade-offs between the performance of radar and communications can be achieved by relaxing the zero-forcing precoder to impose controllable interference levels on the communication systems \cite{6831613}, which offers a more realistic coexistence.
\\\indent More recent contributions have exploited optimization techniques to realize CRSS. In \cite{7485158}, the radar beamformer and communication covariance matrix are jointly designed to maximize the Signal-to-Interference-plus-Noise-Ratio (SINR) of the radar subject to capacity and power constraints at the communication's side. Similar work has been done for the coexistence between the MIMO-matrix completion (MIMO-MC) radar and point-to-point (P2P) MIMO communications \cite{7470514,7953658}, where the radar sub-sampling matrix is further introduced as an optimization variable. To address the more practical coexistence issue between MIMO radar and multi-user MIMO (MU-MIMO) communications, recent work in \cite{7898445} considers the robust beamforming designs with imperfect channel state information (CSI) at the communication's side, where the detection probability of the radar is maximized subject to SINR constraints of the downlink users and the power budget of the BS. As a further development of the technique, a novel beamforming design has been proposed in \cite{liu2017interference} that exploits the interference as a useful power source, which demonstrates orders-of-magnitude  power-savings. While the above coexistence approaches are well-designed, a critical shortfall is that radar and communication devices are required to exchange side-information for achieving a beneficial cooperation, such as the CSI, radar probing waveforms and communication modulation formats. Typically, these exchanges are realized by an all-in-one control center that is connected to both systems via either a wireless link or a backhaul channel \cite{7953658}, which conducts the coordination of the cooperation. In practical scenarios, however, such a control center brings forward considerable extra complexity to the system, and is therefore difficult to implement.
\\\indent In contrast to the above coexistence schemes, a more favorable approach for CRSS is to design a novel dual-functional system that carries out both radar and communications, where the above problem does not exist. Note that such methods are distinctly different from the classic cognitive radio based techniques, as they require the use of specific radar constraints and designs. Recent information theoretical work has shown great potential\cite{7279172,7855671}, but it remains to be seen what benefits can be implemented in practice. As an enabling solution, dual-functional waveform design can support target detection while carrying information at the same time. Such possibilities have been explored for single-antenna systems, where several integrated waveforms that combine the radar and communication signals have been proposed \cite{4268440,5776640,7409935}. Nevertheless, all of these schemes lead to performance loss for either radar or communication, e.g., high peak-average-power-ratio (PAPR) and limited dynamic range \cite{5776640}. As a step further, recent works consider dual-functional waveform design for MIMO systems. In \cite{7347464}, a transmit beampattern for MIMO radar is designed to embed the information bits in sidelobe levels. Related works consider waveform shuffling across the antennas or Phase Shift Keying (PSK) by different beamformer weighting factors as the communication modulation schemes \cite{7575457,7485316}. It should be noted that in the above approaches, one communication symbol is represented by one or several radar pulses, which leads to a low date rate in the order of the radar pulse repetition frequency (PRF). To support multi-user transmission for the cellular downlink, the previous work \cite{liu2017mu} develops a series of beamforming approaches for dual-functional RadCom systems, which will not affect the original modulation scheme and the data rate of the communication system. Nevertheless, the beamforming approaches only focus on the average power constraints, and do not address the design of the constant modulus signals.
\\\indent As an important requirement for both radar and communication applications, the utilization of constant modulus waveforms can avoid signal distortion when the low-cost non-linear power amplifiers are used \cite{5978417}, which leads to an energy-efficient transmission. Such topics have been widely studied for massive MIMO communication scenarios \cite{6451071,7811286,7738555,8023970} as well as the MIMO radar waveform designs \cite{4668417,6649991,7450660,7840054}, where optimization problems with non-convex constant modulus constraint (CMC) are formulated. Due to the NP-hardness of these problems, only suboptimal solutions can be obtained via either convex relaxation methods or local algorithms, such as Semidefinite Relaxation (SDR) \cite{4668417,6649991} and Riemannian manifold methods \cite{7811286,8023970}. Recent MIMO radar work proposes to approach the constant modulus solution by a successive Quadratic Constrained Quadractic Programming (QCQP) Refinement (SQR) procedure \cite{7450660}. Nevertheless, this technique still only guarantees the local optimality of the obtained solution. To the best of our knowledge, the efficient global algorithm for constant modulus waveform design is widely unexplored in the existing literature.
\\\indent In this paper, we propose several optimization-based waveform designs for the dual-functional RadCom systems. It is worth highlighting that all the proposed methods yield probably globally optimal waveforms, which can be used for both target detection and downlink communications. Throughout the paper, we aim to minimize the downlink multi-user interference (MUI) under radar-specific constraints. First, we consider an orthogonal waveform design, which is often used for the initial omnidirectional probing by MIMO radar. Based on this waveform, we extend to the design of a directional radar beampattern that points to the targets of interest. The aforementioned two optimization problems are non-convex, but the optimal solutions can be readily obtained in closed-forms. Still, the obtained performance for the communication system is limited. To allow a trade-off between radar and communication performance, we consider a weighted optimization under non-convex power budget constraint, and obtain its global optimum via a well-designed low-complexity algorithm. Given that both radar and communication systems require constant modulus signals for power-efficient transmission, we finally consider a more practical optimization by enforcing constant modulus constraints and similarity constraints (SC) on the waveform design. In contrast to the existing approaches in both radar and communication works that obtain the local minimizers of problems with CMC \cite{6451071,7811286,7738555,8023970,4668417,6649991,7450660,7840054}, we propose a branch-and-bound method that can efficiently yield a globally optimal solution for the problem. Our numerical results show that the proposed branch-and-bound algorithm considerably outperforms the conventional SQR method\cite{7450660}. For clarity, we summarize our contributions as follows:
\begin{itemize}
    \item We propose dual-functional waveform design approaches for both omnidirectional and directional radar beampatterns, and derive the closed-form solutions.
    \item We propose a weighted optimization that achieves a flexible trade-off between the radar and communication performance, and solve the problem with a low-complexity global algorithm.
    \item We consider the waveform design with CMC and SC constraints, and develop a branch-and-bound algorithm to obtain the globally optimal solutions, which outperforms the conventional SQR algorithm.
    \item We derive the computational complexity for the proposed algorithms analytically.
\end{itemize}

The remainder of this paper is organized as follows, Section II introduces the system model, Section III proposes the closed-form waveform optimizations for radar beampattern design, Section IV considers the trade-off design between radar and communications, Section V solves the problem with CMC and SC constraints, Section VI provides numerical results, and finally Section VII concludes the paper.
\\\indent {\emph{Notations}}: Unless otherwise specified, matrices are denoted by bold uppercase letters (i.e., $\mathbf{H}$), vectors are represented by bold lowercase letters (i.e., $\mathbf{x}$), and scalars are denoted by normal font (i.e., $\rho$). Subscripts indicate the location of the entry in the matrices or vectors (i.e., $s_{i,j}$ and $l_n$ are the $(i,j)$-th and the \emph{n}-th element in $\mathbf{S}$ and $\mathbf{l}$, respectively). $\operatorname{tr}\left(\cdot\right)$, $\left(\cdot\right)^T$, $\left(\cdot\right)^H$ and $\left(\cdot\right)^*$ stand for trace, transpose, Hermitian transpose and complex conjugate, respectively. $\operatorname{Re}\left(\cdot\right)$ and $\operatorname{Im}\left(\cdot\right)$ denote the real and imaginary part of the argument, $\left\| \cdot\right\|$,  $\left\| \cdot\right\|_{\infty}$ and $\left\| \cdot\right\|_F$ denote the $l_2$ norm, $l_{\infty}$ and the Frobenius norm respectively.


\begin{figure}
    \centering
    \includegraphics[width=0.8\columnwidth]{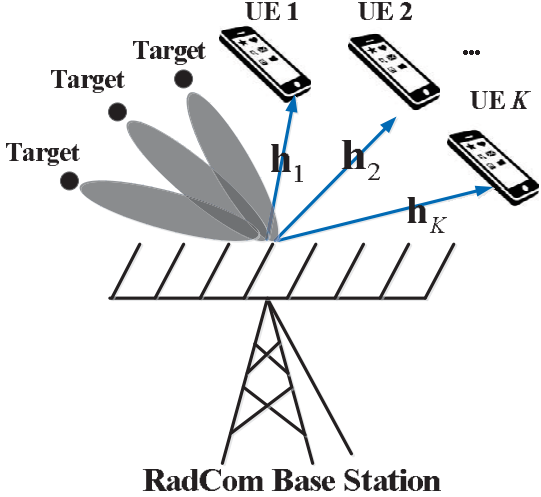}
    \caption{Dual-functional Radar-Communication System.}
    \label{fig:1}
\end{figure}

\section{System Model}
We consider a dual-functional MIMO RadCom system, which simultaneously transmits radar probing waveforms to the targets and communication symbols to the downlink users. The joint system is equipped with a uniform linear array (ULA) with \emph{N} antennas, serving \emph{K} single-antenna users while detecting radar targets at the same time.
\subsection{Communication Model}
The received symbol matrix at the downlink users can be given as
\begin{equation}
    {{\mathbf{Y}}} = {\mathbf{H}}{\mathbf{X}} + {\mathbf{W}},
\end{equation}
where $\mathbf{H} = {\left[ {{{\mathbf{h}}_1},{{\mathbf{h}}_2},...,{{\mathbf{h}}_K}} \right]^T}\in {\mathbb{C}^{K \times N}}$ is the channel matrix, ${\mathbf{X}} = \left[ {{{\mathbf{x}}_1},{{\mathbf{x}}_2},...,{{\mathbf{x}}_L}} \right] \in {\mathbb{C}^{N \times L}}$ is the transmitted signal matrix, with $L$ being the length of the communication frame, ${\mathbf{W}} = \left[ {{{\mathbf{w}}_1},{{\mathbf{w}}_2},...,{{\mathbf{w}}_L}} \right] \in {\mathbb{C}^{K \times L}}$ is the noise matrix, with ${{\mathbf{w}}_j}\sim\mathcal{C}\mathcal{N}\left( {0,{N_0}{\mathbf{I}}_N} \right),\forall j$.
\\\indent Following \cite{liu2017mu}, we rely on the following assumptions: 1) The transmitted signal matrix ${\mathbf{X}}$ is used as dual-functional waveform for both radar and communication operations. In this case, each communication symbol is also a snapshot of a radar pulse; 2) The downlink channel $\mathbf{H}$ is flat Rayleigh fading, and remains unchanged during one communication frame/radar pulse; 3) The channel $\mathbf{H}$ is assumed to be perfectly estimated by pilot symbols.
\\\indent Given the desired constellation symbol matrix ${\mathbf{S}}\in {\mathbb{C}^{K \times L}}$ for the downlink users, the received signals can be rewritten as
\begin{equation}
    {{\mathbf{Y}}} = {\mathbf{S}} + \underbrace {\left( {{\mathbf{HX}} - {\mathbf{S}}} \right)}_{{\text{MUI}}} + {\mathbf{W}},
\end{equation}
For each user, the entry of ${\mathbf{S}}$ is assumed to be drawn from the same constellation. The second term in (2) represents the MUI signals. The total MUI energy can be measured as
\begin{equation}
    {P_{{\text{MUI}}}} = \left\| {{\mathbf{HX}} - {\mathbf{S}}} \right\|_F^2.
\end{equation}
It has been proven in \cite{6451071} that the MUI energy in (3) directly links to the achievable sum-rate of the downlink users. For the \emph{i}-th user, the SINR per frame is given as \cite{6451071}
\begin{equation}
   {\gamma _i} = \frac{{\mathbb{E}\left( {{{\left| {{s_{i,j}}} \right|}^2}} \right)}}{{\underbrace {\mathbb{E}\left( {{{\left| {{\mathbf{h}}_i^T{{\mathbf{x}}_j} - {s_{i,j}}} \right|}^2}} \right)}_{{\text{MUI energy}}} + {N_0}}},
\end{equation}
where ${s_{i,j}}$ is the $\left(i,j\right)$-th entry of $\mathbf{S}$, $\mathbb{E}$ denotes the ensemble average with respect to the time index. It follows that the achievable sum-rate of the users can be given as
\begin{equation}
    R = \sum\limits_{i = 1}^K {{{\log }_2}\left( {1 + {\gamma _i}} \right)}.
\end{equation}
For a given constellation with fixed energy, the power of the useful signal ${\mathbb{E}\left( {{{\left| {{s_{i,j}}} \right|}^2}} \right)}$ is also fixed. Hence, the sum-rate can be maximized by minimizing the MUI energy.
\subsection{Radar Model}
It is widely known that by employing uncorrelated waveforms, MIMO radar achieves higher Degrees of Freedom (DoFs) than the traditional phased-array radar \cite{li2009mimo,4350230}. The existing literature indicates that the design of such a beampattern is equivalent to designing the covariance matrix of the probing signals, where convex optimization can be employed. We refer readers to \cite{4350230,4276989,4516997} for more details on this topic. Here we focus on designing the dual-functional waveform matrix ${\mathbf{X}}$, which has the following spatial covariance matrix
\begin{equation}
    {{\mathbf{R}}_X} = \frac{1}{L}{\mathbf{X}}{{\mathbf{X}}^H}.
\end{equation}
To ensure that ${{\mathbf{R}}_X}$ is positive-definite, we assume $L \ge N$ without loss of generality. Further, the transmit beampattern for the RadCom system can be given as
\begin{equation}
    {P_d}\left( \theta  \right) = {{\mathbf{a}}^H}\left( \theta  \right){{\mathbf{R}}_X}{\mathbf{a}}\left( \theta  \right),
\end{equation}
where $\theta$ denotes the detection angle,  ${\mathbf{a}}\left( \theta \right) = \left[ {1,{e^{j2\pi \Delta \sin \left( \theta  \right)}},...,{e^{j2\pi \left( {N - 1} \right)\Delta \sin \left( \theta  \right)}}} \right]^T \in {\mathbb{C}^{N \times 1}}$ is the steering vector of the transmit antenna array with $\Delta$ being the spacing between adjacent antennas normalized by the wavelength.
\\\indent In the following, we formulate optimization problems that minimize ${P_{{\text{MUI}}}}$ under MIMO radar-specific constraints.

\section{Closed-form Waveform Design for Given Radar Beampatterns}
In this section, we first consider the omnidirectional beampattern design, which is usually used in MIMO radar for initial probing. After that, we consider a directional beampattern design that points to the directions of interest.
\subsection{Omnidirectional Beampattern Design}
For an omnidirectional beampattern, the transmitted waveform matrix ${\mathbf{X}}$ has to be orthogonal, i.e., the corresponding covariance matrix must be the identity matrix. To minimize the MUI energy, the optimization problem is formulated as
\begin{equation}
\begin{gathered}
  \mathop {\min }\limits_{\mathbf{X}} \left\| {{\mathbf{HX}} - {\mathbf{S}}} \right\|_F^2 \hfill \\
  s.t.\;\;\frac{1}{L}{\mathbf{X}}{{\mathbf{X}}^H} = \frac{{{P_T}}}{N}{\mathbf{I}}_N, \hfill \\
\end{gathered}
\end{equation}
where $P_T$ is the total transmit power, ${\mathbf{I}}_N$ denotes the $N\times N$ identity matrix. Problem (8) is obviously non-convex due to the equality constraint, which indicates that ${\mathbf{X}}$ is a point on the Stiefel manifold. Fortunately, it has been proven that (8) can be classified as the so-called Orthogonal Procrustes problem (OPP), which has a simple closed-form global solution based on the Singular Value Decomposition (SVD), and is given as \cite{viklands2006algorithms}
\begin{equation}
    {\mathbf{X}} = \sqrt {\frac{{L{P_T}}}{N}} {\mathbf{U}}{{\mathbf{I}}_{N \times L}}{{\mathbf{V}}^H},
\end{equation}
where ${\mathbf{U\Sigma }}{{\mathbf{V}}^H} = {{\mathbf{H}}^H}{\mathbf{S}}$ is the SVD of ${{\mathbf{H}}^H}{\mathbf{S}}$ with $\mathbf{U}\in {\mathbb{C}^{N \times N}}$ and $\mathbf{V} \in {\mathbb{C}^{L \times L}}$ being the unitary matrices, ${\mathbf{I}}_{N \times L}$ is an $N\times L$ rectangular matrix composed by an $N\times N$ identity matrix and an $N\times\left(L-N\right)$ zero matrix.
\subsection{Directional Beampattern Design}
Given a covariance matrix ${{{\mathbf{R}}_d}}$ that corresponds to a well-designed MIMO radar beampattern, the MUI minimization problem is given as
\begin{equation}
\begin{gathered}
  \mathop {\min }\limits_{\mathbf{X}} \left\| {{\mathbf{HX}} - {\mathbf{S}}} \right\|_F^2 \hfill \\
  s.t.\;\;\frac{1}{L}{\mathbf{X}}{{\mathbf{X}}^H} = {\mathbf{R}}_d, \hfill \\
\end{gathered}
\end{equation}
where ${\mathbf{R}}_d$ is the desired Hermitian positive semidefinite covariance matrix, and $\operatorname{tr} \left( {{{\mathbf{R}}_d}} \right) = {P_T}$. We consider its Cholesky decomposition, which is
\begin{equation}
    {{\mathbf{R}}_d} = {\mathbf{F}}{{\mathbf{F}}^H},
\end{equation}
where ${\mathbf{F}}\in {\mathbb{C}^{N \times N}}$ is a lower triangular matrix. Without loss of generality, we assume ${{\mathbf{R}}_d}$ is positive-definite to ensure that ${\mathbf{F}}$ is invertible. Hence, the constraint in (10) can be equivalently written as
\begin{equation}
    \frac{1}{L}{{\mathbf{F}}^{ - 1}}{\mathbf{X}}{{\mathbf{X}}^H}{{\mathbf{F}}^{ - H}} = {\mathbf{I}}_N.
\end{equation}
Denoting ${\mathbf{\tilde X}} = \sqrt {\frac{1}{L}} {{\mathbf{F}}^{ - 1}}{\mathbf{X}}$, problem (10) can be reformulated as
\begin{equation}
\begin{gathered}
  \mathop {\min }\limits_{{\mathbf{\tilde X}}} \left\| {\sqrt L {\mathbf{HF\tilde X}} - {\mathbf{S}}} \right\|_F^2 \hfill \\
  s.t.\;\;{\mathbf{\tilde X}}{{{\mathbf{\tilde X}}}^H} = {\mathbf{I}}_N, \hfill \\
\end{gathered}
\end{equation}
which is again an OPP problem, and its globally optimal solution is given by
\begin{equation}
    {\mathbf{\tilde X}} = {\mathbf{\tilde U}}{{\mathbf{I}}_{N \times L}}{{{\mathbf{\tilde V}}}^H},
\end{equation}
where ${\mathbf{\tilde U\tilde \Sigma }}{{{\mathbf{\tilde V}}}^H} = {{\mathbf{F}}^H}{{\mathbf{H}}^H}{\mathbf{S}}$ is the SVD of ${{\mathbf{F}}^H}{{\mathbf{H}}^H}{\mathbf{S}}$. It follows that the solution of the original problem (10) is given as
\begin{equation}
    {\mathbf{X}} = \sqrt L {\mathbf{F\tilde U}}{{\mathbf{I}}_{N \times L}}{{{\mathbf{\tilde V}}}^H}.
\end{equation}
\subsection{Complexity Analysis}
The omnidirectional beampattern design includes two matrix multiplications and one SVD, which needs a total of $\mathcal{O}\left(NKL+2NL^2\right)$ complex floating-point-operations (flops), where one complex flop is defined as one complex addition or multiplication. The directional beampattern design, which needs one Cholesky decomposition, four matrix multiplications and one SVD, has the total complexity of $\mathcal{O}\left(2NL^2+N^2L+NKL+N^3+N^2K\right)$. For the conventional communication-only zero-forcing (ZF) precoding, which involves one pseudoinverse for $\mathbf{H}$, and one matrix multiplication between the precoder and the transmitted symbol matrix, the complexity is $\mathcal{O}\left(NKL+N^2K\right)$. It is worth noting that the computational costs of the proposed closed-form approaches share the same order of magnitude with that of the zero-forcing precoder.

\section{Trade-off Between Radar and Communication Performance}
It should be highlighted that both problem (8) and (10) enforce a strict equality constraint, in which case the radar performance is guaranteed to be optimal while the communication counterpart may suffer from serious performance loss. This is particularly pronounced in the cases that the covariance matrices of the communication channel are ill conditioned, where the resulting MUI minimum is still high. We therefore consider a trade-off design by allowing a tolerable mismatch between the designed and the desired radar beampatterns.
\subsection{Problem Formulation}
Let us first denote the optimal solution obtained from (8) or (10) as ${\mathbf{X}}_0$. Given ${\mathbf{X}}_0$, the trade-off problem can be then formulated as
\begin{equation}
\begin{gathered}
  \mathop {\min }\limits_{\mathbf{X}} \rho \left\| {{\mathbf{HX}} - {\mathbf{S}}} \right\|_F^2 + \left( {1 - \rho } \right)\left\| {{\mathbf{X}} - {{\mathbf{X}}_0}} \right\|_F^2 \hfill \\
  s.t.\;\;\frac{1}{L}\left\| {\mathbf{X}} \right\|_F^2 = {P_T}, \hfill \\
\end{gathered}
\end{equation}
where $0 \le \rho  \le 1$ is a weighting factor that determines the weights for radar and communication performance in the dual-functional system. For coherence between (16) and the previous two problems, we enforce an equality constraint for the power budget, as the radar is often required to transmit at its maximum available power in practice.
\\\indent Note that the two Frobenius norms in the objective function can be combined in the form
\begin{equation}
\begin{gathered}
  \rho \left\| {{\mathbf{HX}} - {\mathbf{S}}} \right\|_F^2 + \left( {1 - \rho } \right)\left\| {{\mathbf{X}} - {{\mathbf{X}}_0}} \right\|_F^2 \hfill \\
   = \left\| {{{\left[ {\sqrt \rho  {{\mathbf{H}}^T},\sqrt {1 - \rho } {\mathbf{I}}_N} \right]}^T}{\mathbf{X}} - {{\left[ {\sqrt \rho  {{\mathbf{S}}^T},\sqrt {1 - \rho } {\mathbf{X}}_0^T} \right]}^T}} \right\|_F^2. \hfill \\
\end{gathered}
\end{equation}
Let us denote ${\mathbf{A}} = {\left[ {\sqrt \rho  {{\mathbf{H}}^T},\sqrt {1 - \rho } {\mathbf{I}}_N} \right]^T\in {\mathbb{C}^{\left( K+N\right) \times N}}},{\mathbf{B}} = {\left[ {\sqrt \rho  {{\mathbf{S}}^T},\sqrt {1 - \rho } {\mathbf{X}}_0^T} \right]^T}\in {\mathbb{C}^{\left( K+N\right) \times L}}$, problem (16) can be written compactly as
\begin{equation}
\begin{gathered}
  \mathop {\min }\limits_{\mathbf{X}} \left\| {{\mathbf{AX}} - {\mathbf{B}}} \right\|_F^2 \hfill \\
  s.t.\;\;\left\| {\mathbf{X}} \right\|_F^2 = L{P_T}, \hfill \\
\end{gathered}
\end{equation}
which is a non-convex QCQP, and can be readily transformed into a Semidefinite Programming (SDP) using SDR technique. Since it has only one quadratic constraint, according to \cite{fradkov1979thes,stern1995indefinite}, the SDR is tight, i.e., the solution of the SDR is rank-1, which yields the globally optimal solution of (18). Nevertheless, due to the large number of variables in the problem, SDR is not computationally efficient in general. Hence, we propose a low-complexity algorithm that achieves the global optimum in the following.
\subsection{Low-complexity Algorithm}
We further expand the objective function of (18) as
\begin{equation}
\begin{gathered}
  \left\| {{\mathbf{AX}} - {\mathbf{B}}} \right\|_F^2 = \operatorname{tr} \left( {{{\left( {{\mathbf{AX}} - {\mathbf{B}}} \right)}^H}\left( {{\mathbf{AX}} - {\mathbf{B}}} \right)} \right) \hfill \\
   = \operatorname{tr} \left( {{{\mathbf{X}}^H}{{\mathbf{A}}^H}{\mathbf{AX}}} \right) - \operatorname{tr} \left( {{{\mathbf{X}}^H}{{\mathbf{A}}^H}{\mathbf{B}}} \right) \hfill \\
  \;\;\;\;\;\;\;\;\;\;\;\;\;\;\;\;\;\;\;\;\;\;\;\; - \operatorname{tr} \left( {{{\mathbf{B}}^H}{\mathbf{AX}}} \right) + \operatorname{tr} \left( {{{\mathbf{B}}^H}{\mathbf{B}}} \right). \hfill \\
\end{gathered}
\end{equation}
Defining ${\mathbf{Q}} = {{\mathbf{A}}^H}{\mathbf{A}},{\mathbf{G}} = {{\mathbf{A}}^H}{\mathbf{B}}$, problem (18) can be rewritten as
\begin{equation}
\begin{gathered}
  \mathop {\min }\limits_{\mathbf{X}} \operatorname{tr} \left( {{{\mathbf{X}}^H}{\mathbf{QX}}} \right) - 2\operatorname{Re} \left( {\operatorname{tr} \left( {{{\mathbf{X}}^H}{\mathbf{G}}} \right)} \right) \hfill \\
  s.t.\;\;\left\| {\mathbf{X}} \right\|_F^2 = L{P_T}. \hfill \\
\end{gathered}
\end{equation}
Since $\mathbf{Q}$ is a Hermitian matrix, problem (20) can be viewed as the matrix version of the trust-region subproblem (TRS), for which the strong duality holds\cite{fortin2004trust}, i.e., the duality gap is zero. Let us formulate the Lagrangian multiplier as
\begin{equation}
\begin{gathered}
  \mathcal{L}\left( {{\mathbf{X}},\lambda } \right) = \operatorname{tr} \left( {{{\mathbf{X}}^H}{\mathbf{QX}}} \right) - 2\operatorname{Re} \left( {\operatorname{tr} \left( {{{\mathbf{X}}^H}{\mathbf{G}}} \right)} \right) \hfill \\
  \;\;\;\;\;\;\;\;\;\;\;\;\;\;\;\;\;\;\; + \lambda \left( {\left\| {\mathbf{X}} \right\|_F^2 - L{P_T}} \right), \hfill \\
\end{gathered}
\end{equation}
where $\lambda$ is the dual variable associated with the equality constraint. Let ${\mathbf{X}}_{opt}$ and ${\lambda}_{opt}$ be the primal and dual optimal points with zero duality gap, the optimality conditions for the above TRS can be given as \cite{rendl1997semidefinite}
\begin{subequations}
\begin{align}
  \nabla \mathcal{L}\left( {{{\mathbf{X}}_{opt}},{\lambda_{opt}}} \right) = 2\left( {{\mathbf{Q}} + {\lambda_{opt}}{\mathbf{I}}_N} \right){{\mathbf{X}}_{opt}} - 2{\mathbf{G}} = 0,& \hfill \\
  \left\| {{{\mathbf{X}}_{opt}}} \right\|_F^2 = L{P_T},& \hfill \\
  {\mathbf{Q}} + {\lambda _{opt}}{\mathbf{I}}_N \succeq 0,&
\end{align}
\end{subequations}
where (22b) and (22c) guarantee the primal and the dual feasibility respectively. It follows from (22a) that
\begin{equation}
    {\mathbf{X}}_{opt} = {\left( {{\mathbf{Q}} + \lambda_{opt} {\mathbf{I}}_N} \right)^\dag }{\mathbf{G}},
\end{equation}
where $\left(\cdot\right)^\dag$ denotes the Moore-Penrose pseudoinverse of the matrix. Based on (22b) and (22c) we have
\begin{equation}
\begin{gathered}
\left\| {\left( {{\mathbf{Q}} + {\lambda _{opt}}{\mathbf{I}}_N} \right)^\dag{\mathbf{G}}} \right\|_F^2 \hfill \\
\;\;\;\;\;\;\;\;\;= \left\| {{{\mathbf{V}}}{{\left( {{\mathbf{\Lambda }} + {\lambda _{opt}}{\mathbf{I}}_N} \right)}^{ - 1}}{\mathbf{V}}^H{\mathbf{G}}} \right\|_F^2 =  L{P_T}, \hfill \\
{\lambda _{opt}} \ge  - {\lambda _{\min }}.
\end{gathered}
\end{equation}
where ${\mathbf{Q}} = {\mathbf{V\Lambda }}{{\mathbf{V}}^H}$ is the eigenvalue decomposition of ${\mathbf{Q}}$ with ${\mathbf{V}}$ and ${\mathbf{\Lambda}}$ being the orthogonal and diagonal matrices that contain the eigenvectors and eigenvalues respectively, and $\lambda_{\min}$ is the minimum eigenvalue of $\mathbf{Q}$. One can further show that there exists an unique solution for the equations (24). Let us define
\begin{equation}
\begin{gathered}
  P\left( \lambda  \right) = \left\| {{\mathbf{V}}{{\left( {{\mathbf{\Lambda }} + \lambda {{\mathbf{I}}_N}} \right)}^{ - 1}}{{\mathbf{V}}^H}{\mathbf{G}}} \right\|_F^2 \\
   \;\;\;\;\;= \sum\limits_{i = 1}^N {\sum\limits_{j = 1}^L {\frac{{{{\left( {{{\left[ {{{\mathbf{V}}^H}{\mathbf{G}}} \right]}_{i,j}}} \right)}^2}}}{{{{\left( {\lambda  + {\lambda _i}} \right)}^2}}}} }, \\
\end{gathered}
\end{equation}
where $\lambda_i$ is the \emph{i}-th eigenvalue of ${\mathbf{Q}}$. It can be seen that $P\left( \lambda  \right)$ is strictly decreasing and convex on the interval ${\lambda} \ge  - {\lambda _{\min }}$, which suggests that ${\lambda _{opt}}$ can be obtained by simple line search methods, e.g., Golden-section search\cite{huyer1999global}. Thanks to the eigenvalue decomposition, in each iteration we only need to calculate the inversion of a diagonal matrix. Once the optimal $\lambda$ is obtained, the optimal solution to (16) can be computed by
(23). For clarity, we summarize the above approach in Algorithm 1.
\renewcommand{\algorithmicrequire}{\textbf{Input:}}
\renewcommand{\algorithmicensure}{\textbf{Output:}}
\begin{algorithm}
\caption{Low-complexity Algorithm for Solving (16)}
\label{alg:A}
\begin{algorithmic}
    \REQUIRE $\mathbf{H},\mathbf{S},{\mathbf{x}}_0$, weighting factor $0 \le \rho \le 1$, $P_T$,
    \ENSURE Global minimizer ${\mathbf{X}}_{opt}$
    \STATE 1. Compute $\mathbf A$, $\mathbf B$, $\mathbf Q$ and $\mathbf G$.
    \STATE 2. Compute the eigenvalue decomposition of $\mathbf Q$, set the searching interval as $\left[-\lambda_{min}, b\right]$, where $b \ge 0$ is a searching upper-bound.
    \STATE 3. Find the optimal solution $\lambda_{opt}$ to (24) using Golden-section search.
    \STATE 4. ${\mathbf{X}}_{opt} = {\left( {{\mathbf{Q}} + \lambda_{opt} {\mathbf{I}}_N} \right)^\dag }{\mathbf{G}}.$
\end{algorithmic}
\end{algorithm}
\subsection{Complexity Analysis}
We end this section by analyzing the complexity of Algorithm 1. The Golden-section search method is known to have linear convergence rate, which finds an $\varepsilon_0$-solution within $\mathcal{O}\left(\log\left(1/\varepsilon_0\right)\right)$ iterations. In each iteration we calculate the value of a 1-dimensional function, which suggests that the complexity of the Golden-section search can be omitted in general. Hence the complexity for Algorithm 1 is domainated by the matrix multiplications, the pseudoinverse and the eigenvalue decomposition. Both of the latter two operations involve the computational costs of $\mathcal{O}\left(N^3\right)$ complex flops, and the matrix multiplications involve the complexity of $\mathcal{O}\left(N^2L+NKL+N^3+N^2K\right)$. Therefore, the total complexity for Algorithm 1 is $\mathcal{O}\left(N^2L+NKL+3N^3+N^2K\right)$, which again shares the same order of magnitude with the communication-only ZF precoding. For the sake of clarity, we summarize the computational costs of the proposed three waveform design approaches in TABLE I.


\begin{table}
\renewcommand{\arraystretch}{1.3}
\caption{Computational Complexity for the Proposed Approaches}
\label{table_example}
\centering
\begin{tabular}{lr}
\toprule
 \bf Method & \bf Complex Flops \\
\midrule
Omnidirectional Design & $\mathcal{O}\left(NKL+2NL^2\right)$  \\
Directional Design & $\mathcal{O}\left(2NL^2+N^2L+NKL+N^3+N^2K\right)$ \\
Trade-off Design & $\mathcal{O}\left(N^2L+NKL+N^3+N^2K\right)$ \\
Zero-forcing (benchmark) & $\mathcal{O}\left(NKL+N^2K\right)$ \\
\bottomrule
\end{tabular}
\end{table}

\section{Constant Modulus Waveform Design}
In the previous sections, the dual-functional RadCom waveform is designed under total power constraints, which is not guaranteed to generate constant modulus signals. In this section, we consider the RadCom waveform design that minimizes the communication MUI energy given the CMC.
\subsection{Problem Formulation}
Following the same notations in the previous section, our optimization problem can be formulated as
\begin{subequations}
\begin{align}
  &\mathop {\min }\limits_{\mathbf{X}} \left\| {{\mathbf{HX}} - {\mathbf{S}}} \right\|_F^2 \hfill \\
  &\;s.t\;\;{\left\| {\operatorname{vec} \left( {{\mathbf{X}} - {{\mathbf{X}}_0}} \right)} \right\|_\infty } \le \eta , \hfill \\
  &\;\;\;\;\;\;\;\left| {{x_{i,j}}} \right| = \sqrt {\frac{{{P_T}}}{N}}, \forall i, j,
\end{align}
\end{subequations}
where ${{\mathbf{X}}_0}\in {\mathbb{C}^{N \times L}}$ is a known benchmark radar signal matrix that has constant-modulus entries, e.g., chirp signals, $\operatorname{vec} \left( \cdot\right)$ denotes the vectorization of a matrix, and ${x_{i,j}}$ is the $\left(i,j\right)$-th entry of ${{\mathbf{X}}}$. The constraint (26b) is called similarity constraint (SC) in the radar literature\cite{6649991}, which controls the difference between the designed waveform and the benchmark with $\eta$ being the tolerable difference.
\\\indent It is trivial to see that the objective function of (26) is separable, since
\begin{equation}
    \left\| {{\mathbf{HX}} - {\mathbf{S}}} \right\|_F^2 = \sum\limits_{i = 1}^L {\left\| {{\mathbf{H}}{{\mathbf{x}}_i} - {{\mathbf{s}}_i}} \right\|^2}.
\end{equation}
Hence, it can be further simplified using the normalized vector variable, which is
\begin{equation}
\begin{gathered}
  \mathop {\min }\limits_{\mathbf{x}} \left\| {{\sqrt {\frac{{{P_T}}}{N}}\mathbf{Hx}} - {\mathbf{s}}} \right\|^2 \hfill \\
  s.t\;\;\;{\left\| {{\mathbf{x}} - {{\mathbf{x}}_0}} \right\|_\infty } \le \varepsilon , \hfill \\
  \;\;\;\;\;\;\;\left| {{x\left(n\right)}} \right| = 1 ,\forall n, \hfill \\
\end{gathered}
\end{equation}
where $\varepsilon = \eta\sqrt {\frac{{{N}}}{P_T}}$, ${\mathbf{x}}\in {\mathbb{C}^{N \times 1}}$, ${\mathbf{x}}_0\in {\mathbb{C}^{N \times 1}}$ are the columns of $\mathbf{X}$ and ${\mathbf{X}}_0$ normalized by $\sqrt {\frac{{{P_T}}}{N}}$, ${\mathbf{s}}\in {\mathbb{C}^{K \times 1}}$ is the column of $\mathbf{S}$, and $x\left(n\right)$ denotes the \emph{n}-th entry of ${\mathbf{x}}$. Since problem (26) can be solved by solving the problem (28) for each column of ${\mathbf{X}}$ concurrently, we will focus on (28) in the following discussion. For notational convenience, we omit the column index.
\\\indent Note that $0 \le \varepsilon \le 2$ since both ${\mathbf{x}}$ and ${\mathbf{x}}_0$ have unit modulus. According to \cite{6649991}, the similarity constraint can be rewritten as
\begin{equation}
    \arg {x\left(n\right)} \in \left[ {{l_n},{u_n}} \right], \forall n,
\end{equation}
where
\begin{equation}
\begin{gathered}
  {l_n} = \arg {x_0}\left( n \right) - \arccos \left( {1 - {{{\varepsilon ^2}} \mathord{\left/
 {\vphantom {{{\varepsilon ^2}} 2}} \right.
 \kern-\nulldelimiterspace} 2}} \right), \hfill \\
  {u_n} = \arg {x_0}\left( n \right) + \arccos \left( {1 - {{{\varepsilon ^2}} \mathord{\left/
 {\vphantom {{{\varepsilon ^2}} 2}} \right.
 \kern-\nulldelimiterspace} 2}} \right), \hfill \\
\end{gathered}
\end{equation}
which leads to the following equivalent formulation of the problem
\begin{equation}
\begin{gathered}
  \mathop {\min }\limits_{\mathbf{x}} f\left({\mathbf{x}}\right) = \left\| {{\mathbf{\tilde H}}\mathbf{x} - {\mathbf{s}}} \right\|^2 \hfill \\
  s.t\;\;\;\arg {x\left(n\right)} \in \left[ {{l_n},{u_n}} \right], \forall n, \hfill \\
  \;\;\;\;\;\;\;\left| {{x\left(n\right)}} \right| = 1 ,\forall n, \hfill \\
\end{gathered}
\end{equation}
where ${\mathbf{\tilde H}} = \sqrt {\frac{{{P_T}}}{N}} {\mathbf{H}}$. For each $x\left(n\right)$, the feasible region is an arc on the unit circle as shown in Fig. 2, which makes the problem non-convex, and NP-hard in general. In the following, we consider a global optimization algorithm for solving (28), which is based on the general framework of the branch-and-bound (BnB) methodology \cite{Tuy2016Convex}.

\begin{figure}
    \centering
    \includegraphics[width=0.8\columnwidth]{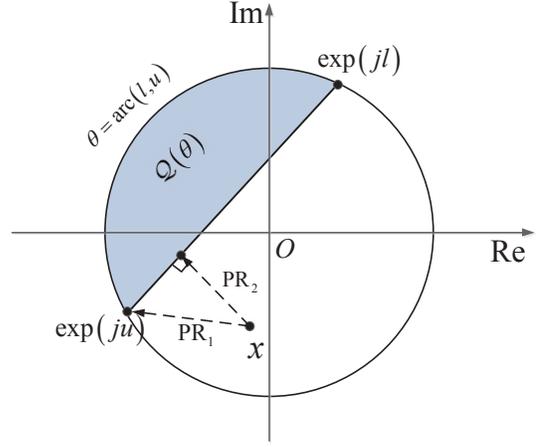}
    \caption{Feasible region and convex hull of problem (28).}
    \label{fig:2}
\end{figure}

\subsection{The Branch-and-bound Framework}
A typical BnB algorithm requires to partition the feasible region into several subregions, where we formulate corresponding subproblems. For each subproblem, we obtain a sequence of asymptotic lower-bounds and upper-bounds by well-designed bounding functions. In each iteration, we update the bounds and the set of the subproblems following the BnB rules until convergence, i.e., the difference between the upper-bound and lower-bound goes to zero.
\\\indent It is well-known that the worst-case complexity for the BnB algorithm is of the exponential order with respect to $N$, i.e., to search all the branches of the subproblems exhaustively \cite{Tuy2016Convex}. Nevertheless, by carefully choosing the tightest bounds, it is possible to efficiently identify and prune the unqualified branches, which accelerates the algorithm significantly.
\\\indent Let us denote the feasible region, i.e., the arc shown in Fig. 2, as ${\theta _n} = \operatorname{arc}\left( {{l_n},{u_n}} \right)$. Problem (31) can be compactly written as
\begin{equation}
\begin{gathered}
  {\mathcal{P}}\left({{\Theta_0}}\right):\;\mathop {\min }\limits_{\mathbf{x}} f\left({\mathbf{x}}\right) \hfill \\
  \;\;\;\;\;\;\;\;\;\;\;\;\;\;\;s.t.\;\;{\mathbf{x}} \in {\Theta_0}. \hfill \\
\end{gathered}
\end{equation}
where ${\Theta_0} = {\theta_1}\times{\theta_2}\times...\times{\theta_N}$, and $f\left({\mathbf{x}}\right)$ is defined in (31). By the above notations, a subproblem can be denoted as ${\mathcal{P}}\left({ \Theta}\right)$, where ${\Theta}\subseteq{\Theta_0}$ is the corresponding subregion. We then find a lower-bound of ${\mathcal{P}}\left({ \Theta}\right)$ by a lower-bounding function
\begin{equation}
    {f_L}\left( {\Theta } \right) = f\left( {{\mathbf{x}}_l} \right),
\end{equation}
where ${{\mathbf{x}}_l}$ is a relaxed solution that achieves the bound. In order to compute the upper-bound, we find a feasible solution ${{{\mathbf{x}}_{u}}}$ for ${\mathcal{P}}\left({ \Theta}\right)$. The upper-bounding function is thus given by
\begin{equation}
    {f_U}\left( {\Theta } \right) = f\left( {\mathbf{x}}_u \right).
\end{equation}
The above bounding functions (33) and (34) will be specified in the next subsection. Here we only use $f_L$ and $f_U$ to introduce the BnB framework for notational convenience. In the BnB algorithm, we store all the subproblems in a problem set ${\mathcal{S}}$, which will be updated together with the global bounds in each iteration by the following rules\cite{Tuy2016Convex}
\begin{enumerate}
    \item {\bf{Branching}}: Pick a problem ${\mathcal{P}}\left({\Theta}\right)\in {\mathcal{S}}$ that yields the smallest lower-bound. Equally divide ${\Theta}$ into two subregions following some subdivision rules detailed in the following, and generate two subproblems. Then delete ${\mathcal{P}}\left({\Theta}\right)$ in the problem set.
    \item {\bf{Pruning}} (optional): Evaluate the qualification of the two subproblems, if their lower-bounds are less than the current global upper-bound, add them into ${\mathcal{S}}$.
    \item {\bf{Bounding}}: Choose the smallest lower-bound and upper-bound from ${\mathcal{S}}$ as the bounds for the next iteration.
\end{enumerate}
Note that the pruning step is only for saving the memory of storing $\mathcal{S}$, and will not affect the effectiveness of the BnB procedure. This is because by choosing the smallest bounds in $\mathcal{S}$ we can always avoid the unqualified branches. For clarity, we summarize our BnB algorithm in Algorithm 2.
\renewcommand{\algorithmicrequire}{\textbf{Input:}}
\renewcommand{\algorithmicensure}{\textbf{Output:}}
\begin{algorithm}
\caption{Branch-and-Bound Method for Solving (28)}
\label{alg:B}
\begin{algorithmic}
    \REQUIRE $\mathbf{\tilde H},\mathbf{S},{\mathbf{x}}_0$, $0 \le \varepsilon \le 2$, tolerance threshold $\delta > 0$, bound functions $f_L$ and $f_U$.
    \STATE {\bf{Initialization}}: Let ${{\Theta}_{0}}$ be the initial feasible region of problem (27), ${{\mathcal{S}}} = \left\{ {{\mathcal{P}}\left({\Theta_{0}}\right), {f_U}\left( {{\Theta_{0}}} \right), {f_L}\left( {{\Theta_{0}}} \right)} \right\}$ be the initialized subproblem set. Set ${UB} = {f_U}\left( {{\Theta_{0}}} \right)$, ${LB} = {f_L}\left( {{\Theta_{0}}} \right)$.
    \WHILE{{$UB-LB > \delta$}}
        \STATE {\bf{Branching}}
        \STATE a) Pick ${{\mathcal{P}}\left( {{\Theta }} \right)}\in {{\mathcal{S}}}$, such that ${f_L}\left( \Theta  \right) = {LB}$. Update ${{\mathcal{S}}} = {{\mathcal{S}}}\backslash {\mathcal{P}}\left( \Theta \right)$.
        \STATE b) Divide ${\Theta}$ into ${\Theta_{A}}$ and ${\Theta_{B}}$ following the chosen subdivision rule.
        \STATE {\bf{Bounding}}
        \STATE a) Compute ${f_U}\left( {{\Theta_{i}}} \right)$ and ${f_L}\left( {{\Theta_{i}}} \right)$ for ${\mathcal{P}}\left({\Theta_{i}}\right), i = A,B,$ and  add them to $\mathcal{S}$.
        \STATE b) Update $UB$ and $LB$ as the smallest upper-bound and lower-bound in ${\mathcal{S}}$, respectively.
    \ENDWHILE
    \ENSURE ${\mathbf{x}}_{opt} = $ the feasible solution that achieves $UB$.
\end{algorithmic}
\end{algorithm}
\\\indent To ensure that Algorithm 2 converges in a finite number of iterations, the chosen subproblem for branching, the subdivision rule and the bounding functions $f_L$ and $f_U$ should satisfy the following conditions \cite{Tuy2016Convex}
\begin{enumerate}
    \item The branching is bounding-improving, i.e., in each iteration we choose the problem that yields the smallest lower-bound as the branching node.
    \item The subdivision is exhaustive, i.e., the maximum length of the subregions converges to zero as the iteration number goes to infinity.
    \item The bounding is consistent with branching, i.e., $UB - f_{opt}$ converges to zero as the maximum length of the subregions goes to zero, where ${ f_{opt}}$ is the optimal value of the original problem.
\end{enumerate}
Our Algorithm 2 satisfies condition 1) automatically. We then choose the subdivision rules to obtain the subproblems from the branching node. For a given node ${\mathcal{P}}\left(\Theta\right)$, we consider the following two rules:
\begin{itemize}
    \item {\bf{Basic rectangular subdivision}} (BRS): Equally divide $\Theta$ along $arc\left(l_n,u_n\right)$ and keep $arc\left(l_i,u_i\right), \forall i \ne n$ unchanged, where
    \begin{equation}
        n = \arg \mathop {\max }\limits_n \left\{ {{\phi _n}\left| {{\phi _n} = {u_n} - {l_n}} \right.} \right\}.
    \end{equation}
    \item {\bf{Adaptive rectangular subdivision}} (ARS): Equally divide $\Theta$ along $arc\left(l_n,u_n\right)$ and keep $arc\left(l_i,u_i\right), \forall i \ne n$ unchanged, where
    \begin{equation}
        n = \arg \mathop {\max }\limits_n \left\{ {{d_n}\left| {{d_n} = \left| {{x_u}\left( n \right) - {x_l}\left( n \right)} \right|} \right.} \right\}.
    \end{equation}
    In (35) ${\mathbf{x}}_u $ and ${\mathbf{x}}_l$ are the solutions associated with ${f_U}\left( {{\Theta}} \right)$ and ${f_L}\left( {{\Theta}} \right)$, respectively.
\end{itemize}
According to [45, Theorem 6.3 and 6.4], both the above two rules satisfy condition 2). In practical simulations, we observe that ARS has a faster convergence rate than BRS.
\subsection{Upper-bound and Lower-bound Acquisition}
It remains to develop approaches to acquire the lower and upper bounds, which are key to accelerating the BnB procedure. Following the approach in \cite{7914783}, we compute the lower-bound by the convex relaxation of (32). As shown in Fig. 2, the convex hull for each entry $x\left(n\right)$, denoted as ${\mathcal{Q}}\left(\theta _n\right)$, is a circular segment, and can be given as
\begin{equation}
    \mathcal{Q}\left( {{\theta _n}} \right):\left\{ {x\left| {\arg \left( x \right) \in {\theta _n},\left| x \right| \le 1} \right.} \right\}.
\end{equation}
By simple analytic geometry, the angle constraint can be equivalently written as
\begin{equation}
    \operatorname{Re} \left( {{x^*}\left( {\frac{{{e^{ju}} + {e^{jl}}}}{2}} \right)} \right) \ge \cos \left( {\frac{{u - l}}{2}} \right),
\end{equation}
which is nothing but a linear constraint. It follows that the constraint for the vector variable is
\begin{equation}
    \operatorname{Re} \left( {{{\mathbf{x}}^*} \circ \left( {\frac{{{e^{j{\mathbf{u}}}} + {e^{j{\mathbf{l}}}}}}{2}} \right)} \right) \ge \cos \left( {\frac{{{\mathbf{u}} - {\mathbf{l}}}}{2}} \right),
\end{equation}
where $\mathbf{u} = {\left[ {{{u}_1},{{u}_2},...,{{u}_N}} \right]^T}\in {\mathbb{R}^{N \times 1}},\mathbf{l} = {\left[ {{{l}_1},{{l}_2},...,{{l}_N}} \right]^T}\in {\mathbb{R}^{N \times 1}}$, and $\circ$ denotes the Hadamard product. Hence, the convex relaxation can be given as the following QCQP problem
\begin{subequations}
\begin{align}
  &\operatorname{QP-LB}:\mathop {\min}\limits_{\mathbf{x}} \;{\left\| {{\mathbf{Hx}} - {\mathbf{s}}} \right\|^2} \hfill \\
  &\;\;\;\;\;\;\;\;\;\;\;\;\;\;\;\;s.t.\operatorname{Re} \left( {{{\mathbf{x}}^*}\circ \left( {\frac{{{e^{j{\mathbf{u}}}} + {e^{j{\mathbf{l}}}}}}{2}} \right)} \right) \ge \cos \left( {\frac{{{\mathbf{u}} - {\mathbf{l}}}}{2}} \right), \hfill \\
  &\;\;\;\;\;\;\;\;\;\;\;\;\;\;\;\;\;\;\;\;\;\;\left|x\left( n \right)\right|^2 \le 1,\forall n.
\end{align}
\end{subequations}
Problem (40) can be efficiently solved via numerical solvers, e.g., the CVX toolbox. By doing so, we can readily obtain the lower-bound for each subproblem.
\\\indent A natural way to compute the upper-bound is to project each entry of the obtained solution ${\mathbf{x}}_{l}$ of (40) on the corresponding arc to get a feasible solution. Such a projector can be given in an element-wise manner as follows
\begin{equation}
{\operatorname{PR}}_1\left( x \right) = \left\{ \begin{gathered}
  {x \mathord{\left/
 {\vphantom {x {\left| x \right|}}} \right.
 \kern-\nulldelimiterspace} {\left| x \right|}},\arg x \in \left[ {l,u} \right], \hfill \\
  \exp \left( {jl} \right),\arg x \in \left[ {{(l+u)}/{2}+\pi,l+2\pi} \right], \hfill \\
  \exp \left( {ju} \right),\arg x \in \left[ {u,{(l+u)}/{2}+\pi} \right], \hfill \\
\end{gathered}  \right.
\end{equation}
where we omit the subscripts for convenience.
\\\indent The upper-bound obtained by the projector (41) is still loose in general. To get a tighter bound, one can use ${\operatorname{PR}}_1\left( {\mathbf{x}}_{l} \right)$ as the initial point, and solve the following non-convex QCQP
\begin{subequations}
\begin{align}
  &\operatorname{QP-UB}:\mathop {\min}\limits_{\mathbf{x}} \;{\left\| {{\mathbf{Hx}} - {\mathbf{s}}} \right\|^2} \hfill \\
  &\;\;\;\;\;\;\;\;\;\;\;\;\;\;\;\;s.t.\operatorname{Re} \left( {{{\mathbf{x}}^*}\circ \left( {\frac{{{e^{j{\mathbf{u}}}} + {e^{j{\mathbf{l}}}}}}{2}} \right)} \right) \ge \cos \left( {\frac{{{\mathbf{u}} - {\mathbf{l}}}}{2}} \right), \hfill \\
  &\;\;\;\;\;\;\;\;\;\;\;\;\;\;\;\;\;\;\;\;\;\;\left|x\left( n \right)\right|^2 = 1,\forall n.
\end{align}
\end{subequations}
which can be locally solved via the \emph{fmincon} solver in MATLAB. Since the solver employs descent methods, the obtained local minimizer is guaranteed to yield a smaller value than $f\left({\operatorname{PR}}_1\left( {\mathbf{x}}_{l} \right)\right)$.
\\\indent To further accelerate the speed for solving QP-LB and obtaining the bounds, we consider accelerated gradient projection (GP) methods \cite{Nesterov2004Introductory} in addition to the QCQP solvers. Given  $x_n \in \mathbb{C}$, the projector ${\operatorname{PR}}_2$ projects $x_n$ to the nearest point in the corresponding convex hull ${\mathcal{Q}}\left(\theta _n\right)$. The details for deriving ${\operatorname{PR}}_2$ are provided in the Appendix. Here we briefly introduce our iterative scheme as
\begin{equation}
  {\mathbf{v}} = {{\mathbf{x}}^{\left( k \right)}} + \frac{{k - 1}}{{k + 2}}\left( {{{\mathbf{x}}^{\left( k \right)}} - {{\mathbf{x}}^{\left( {k - 1} \right)}}} \right),
\end{equation}
\begin{equation}
  {{\mathbf{x}}^{\left( {k + 1} \right)}} = {\operatorname{PR}_2}\left( {{\mathbf{v}} - 2s{{\mathbf{\tilde H}}^H}\left( {{\mathbf{\tilde H}}{\mathbf{v}} - {\mathbf{s}}} \right)} \right),
\end{equation}
where we start from ${\mathbf{x}}^{\left(0\right)}$ and ${\mathbf{x}}^{\left(1\right)} = {\mathbf{x}}^{\left(0\right)}$. For the least-squares objective function, we choose the stepsize as $s = 1/{\tilde \lambda}_{\max}$, where ${\tilde \lambda}_{\max}$ is the maximum eigenvalue of ${\mathbf{\tilde H}}^H{\mathbf{\tilde H}}$, i.e., the Lipschitz constant.
\\\indent Note that the above iteration scheme can only be used for convex feasible regions due to the interpolation operation (43). For the non-convex QP-UB problem (42), we use ${\mathbf{x}}^{\left( k \right)}$ instead of the interpolated point $\mathbf{v}$, and replace the projector ${\operatorname{PR}}_2$ with ${\operatorname{PR}}_1$, which projects the point onto the arc, i.e., the feasible region. Similar to (40), we use ${\operatorname{PR}}_1\left( {\mathbf{x}}_{l} \right)$ as the initial point.
\\\indent Based on \cite{lobo1998applications}, the complexity for using interior-point method to solve the QCQP problems is ${\mathcal{O}}\left(N^{3}\right)$ per iteration. For both gradient-based methods, the costs are ${\mathcal{O}}\left(NK\right)$ in each iteration, which are far more efficient in terms of a fixed iteration number.
\subsection{Convergence Analysis and Worst-case Complexity}
We end this section by analyzing the convergence behavior and the worst-case complexity for the proposed Algorithm 2. The convergence proof is to show that our bounding functions $f_L$ and $f_U$ satisfy the condition 3). Recall the definitions of $\phi_n$ and $d_n$ in (35) and (36).
By denoting $\phi _{\max } = \max \left\{ {\phi _n} \right\}, d _{\max } = \max \left\{ {d _n} \right\}$, we have the following Lemma 1.
\begin{lemma}
As $\phi_{\max}$ or $d_{\max}$ goes to zero, the difference between $UB$ and $LB$ uniformly converges to zero, i.e.,
\begin{equation}
\begin{gathered}
  \forall \delta  > 0,\exists \eta_1, \eta_2 \ge 0 \;\;s.t.\; \hfill \\
  {\phi _{\max }} \le \eta_1 \;\;\text{or}\;\; {d _{\max }} \le \eta_2  \Rightarrow {UB} - {LB} \le \delta . \hfill \\
\end{gathered}
\end{equation}
\end{lemma}
\begin{proof}
Let us first denote the points that generate $UB$ and $LB$ as ${\mathbf{x}}_u$ and ${\mathbf{x}}_l$, i.e., $UB = f\left({\mathbf{x}}_u\right), LB = f\left({\mathbf{x}}_l\right)$. Following the Lagrange Mean-value Theorem we have
\begin{equation}
\begin{gathered}
  {UB} - {LB} = f\left( {{{\mathbf{x}}_u}} \right) - f\left( {{{\mathbf{x}}_l}} \right) \hfill \\
   = \nabla {f^H}\left( {\mathbf{z}} \right)\left( {{{\mathbf{x}}_u} - {{\mathbf{x}}_l}} \right) \hfill \\
   \le \left\| {\nabla f\left( {\mathbf{z}} \right)} \right\|\left\| {\left( {{{\mathbf{x}}_u} - {{\mathbf{x}}_l}} \right)} \right\|, \hfill \\
\end{gathered}
\end{equation}
where
\begin{equation}
    {\mathbf{z}} \in \left\{ {{\mathbf{w}}\left| {{\mathbf{w}} = t{{\mathbf{x}}_u} + \left( {1 - t} \right){{\mathbf{x}}_l},t \in \left[ {0,1} \right]} \right.} \right\}.
\end{equation}
The upper-bound of the gradient is given as
\begin{equation}
\begin{gathered}
  \left\| {\nabla f\left( {\mathbf{z}} \right)} \right\| = 2\left\| {{{\mathbf{\tilde H}}^H}\left( {{\mathbf{\tilde H}}{\mathbf{z}} - {\mathbf{s}}} \right)} \right\| \hfill \\
   \le 2\left\| {{{\mathbf{\tilde H}}^H}{\mathbf{\tilde H}}{\mathbf{z}}} \right\| + 2\left\| {{{\mathbf{\tilde H}}^H}{\mathbf{s}}} \right\| \hfill \\
   \le 2\sqrt N {{\tilde \lambda} _{\max }} + 2\left\| {{{\mathbf{\tilde H}}^H}{\mathbf{s}}} \right\|, \hfill \\
\end{gathered}
\end{equation}
where the second line of (48) is based on the triangle inequality, the third line is based on the definition of the matrix $l_2$ norm.
\\\indent For the convex hull of each $\operatorname{arc}\left(l_n,u_n\right)$, the longest line segment is the chord shown in Fig. 2 ($\phi_n \le \pi$) or the diameter ($\phi_n \ge \pi$). By simple geometric relations we have
\begin{equation}
    \left\| {{{\mathbf{x}}_u} - {{\mathbf{x}}_l}} \right\| \le \sqrt N {d_{\max }} \le 2\sqrt {\sum\limits_{n = 1}^N {{{\sin }^2}\left( {\frac{{\min \left( {{\phi _n},\pi } \right)}}{2}} \right)} } .
\end{equation}
For $\phi_n \le \pi ,\forall n$, it follows that
\begin{equation}
    \left\| {{{\mathbf{x}}_u} - {{\mathbf{x}}_l}} \right\| \le \sqrt N {d_{\max }} \le 2\sqrt N \sin \left( {\frac{{{\phi _{\max }}}}{2}} \right).
\end{equation}
By using (46), (48) and (50) we obtain
\begin{equation}
{UB} - {LB} \le 4\left( {N{{\tilde \lambda} _{\max }} + \sqrt N \left\| {{{\mathbf{\tilde H}}^H}{\mathbf{s}}} \right\|} \right)\sin \left( {\frac{{{\phi _{\max }}}}{2}} \right),
\end{equation}
\begin{equation}
UB - LB \le 2\left( {N{{\tilde \lambda} _{\max }} + \sqrt N \left\| {{{\mathbf{\tilde H}}^H}{\mathbf{s}}} \right\|} \right){d_{\max }}.
\end{equation}
Given any $\delta > 0$, let
\begin{equation}
    \eta_1  = \min \left( {\pi ,2\arcsin \left( {\frac{\delta }{{4\left( {N{{\tilde \lambda} _{\max }} + \sqrt N \left\| {{{\mathbf{\tilde H}}^H}{\mathbf{s}}} \right\|} \right)}}} \right)} \right),
\end{equation}
\begin{equation}
    {\eta _2} = \frac{\delta }{{2\left( {N{{\tilde \lambda} _{\max }} + \sqrt N \left\| {{{\mathbf{\tilde H}}^H}{\mathbf{s}}} \right\|} \right)}},
\end{equation}
we have ${UB} - {LB} \le \delta$ if ${\phi}_{\max} \le \eta_1$ or $d_{\max} \le \eta_2$.
\end{proof}
\begin{thm}
Algorithm 2 converges in a finite number of iterations to a value arbitrary close to ${ f_{opt}}$.
\end{thm}
\begin{proof}
Algorithm 2 satisfies both conditions 1) and 2). Furthermore, according to the definition of $UB$ and $LB$ we have
\begin{equation}
    0 \le UB-{{ f_{opt}}} \le UB - LB.
\end{equation}
According to Lemma 1, the bounding is consistent with branching for the proposed two subdivision rules. Therefore, Algorithm 2 satisfies condition 3), which completes the proof.
\end{proof}
The following Theorem 2 specifies the worst-case complexity of Algorithm 2 for using BRS.
\begin{thm}
When the BRS is used, Algorithm 2 converges to a $\delta$-optimal solution for at most
\begin{equation}
T = \left\lceil {\frac{{{2^{N + 1}}{{\arccos }^N}\left( {1 - {\varepsilon ^2}/2} \right)}}{\eta_1 }} \right\rceil
\end{equation}
iterations, where $\eta_1$ is given by (53).
\end{thm}
\begin{proof}
Define $\operatorname{vol} \left( {{\Theta _0}} \right) = {\left( {2\arccos \left( {1 - {\varepsilon ^2}/2} \right)} \right)^N}$ as the volume of the initialized feasible region. Assume that Algorithm 2 terminates at the \emph{T}-th iteration. According to \cite{6119233}, we have
\begin{equation}
    \frac{\phi_{\max}}{2} \le \frac{{\operatorname{vol} \left( {{\Theta _{0}}} \right)}}{T} \le \frac{\eta_1}{2}.
\end{equation}
It follows that
\begin{equation}
    T \ge \frac{{2\operatorname{vol} \left( {{\Theta _{0}}} \right)}}{\eta_1} = \frac{{{2^{N + 1}}{{\arccos }^N}\left( {1 - {\varepsilon ^2}/2} \right)}}{\eta_1},
\end{equation}
which completes the proof.
\end{proof}
Although the worst-case costs of both BnB-ARS and BnB-BRS are at the exponential order with $N$, our simulations show that in most cases, the algorithm terminates at a small iteration number thanks to the tight bounds.

\section{Numerical Results}
In this section, we present numerical results to validate the effectiveness of the proposed waveform design approaches. For convenience, we set $P_T = 1$, and each entry of the channel matrix $\mathbf{H}$ subject to standard Complex Gaussian distribution, i.e., $h_{i,j} \sim\mathcal{C}\mathcal{N}\left( {0,1} \right) $. In all the simulations, we set $N = 16$ and employ a ULA with half-wavelength spacing between the adjacent antennas. The constellation chosen for the communication users is the unit-power QPSK alphabet, i.e., the power of each entry in the symbol matrix $\mathbf{S}$ is 1.
\subsection{Dual-functional Waveform Design with Given Radar Beampatterns}
\begin{figure}
    \centering
    \includegraphics[width=\columnwidth]{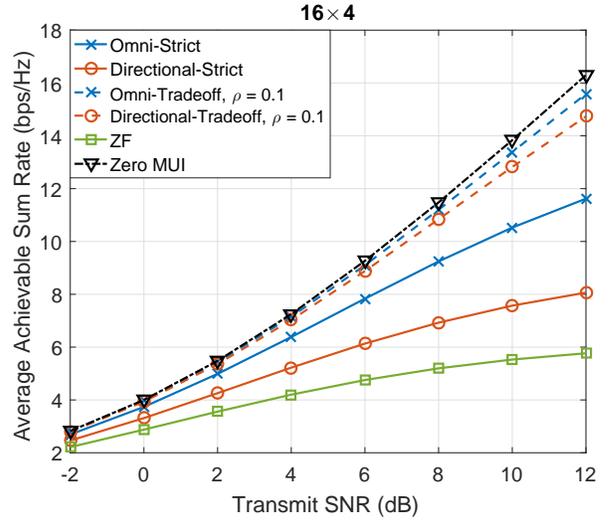}
    \caption{Sum-rate comparison for different approaches, $N = 16, K = 4$.}
    \label{fig:3}
\end{figure}
\begin{figure}
    \centering
    \includegraphics[width=\columnwidth]{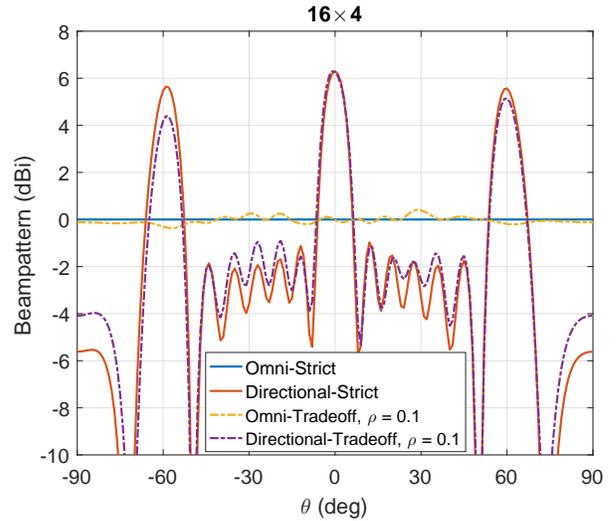}
    \caption{Radar beampatterns obtained by different approaches, $N = 16, K = 4$.}
    \label{fig:4}
\end{figure}
We first show the communication sum-rate obtained by different approaches as well as the associated radar beampatterns in Fig. 3 and Fig. 4, respectively, where we define $\operatorname{SNR} = {P_T}/{N_0}$, and use `Omni', `Directional' and `ZF' to represent omnidirectional beampattern design, directional beampattern design and zero-forcing precoding based on the problems (8) and (10). Further, we denote the waveform design with strict equality constraints (8) and (10) and the trade-off design (16) as `Strict' and `Tradeoff' respectively. The length of the communication frame/radar pulse is set as $L = 20$. For directional beampattern design, we consider three targets of interest with angles of $-\pi/3, 0$ and $\pi/3$, and exploit the classic Least-Squares techniques \cite{4516997} to obtain the desired covariance matrix ${\mathbf{R}}_d$ as defined in (10). It can be observed in Fig.3 that, the proposed two strict waveform designs outperform the communication-only zero-forcing precoding significantly, despite that their computational costs remain at the same level as we have discussed in Section III. The resultant radar beampatterns are shown in Fig. 4 for `Strict', which are exactly the same with the desired beampatterns. Moreover, by introducing a small weighting factor $\rho = 0.1$ to the communication side, the sum-rates for trade-off designs increase significantly by approaching to that of the zero MUI case, i.e., the AWGN capacity. Meanwhile in Fig. 4, the corresponding radar beampatterns only experience slight performance-loss.
\begin{figure}
    \centering
    \includegraphics[width=\columnwidth]{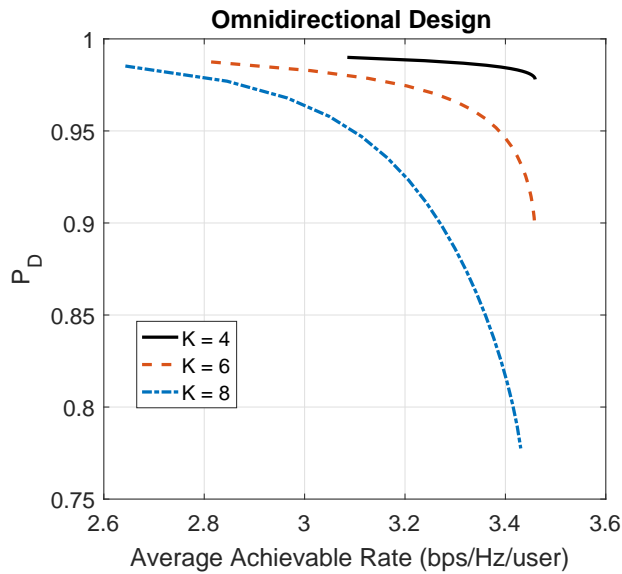}
    \caption{Trade-off between the achievable rate per user and the radar detection probability for omnidirectional beampattern design, $N = 16, \text{radar}\;\;{\text{SNR}} = -6\text{dB}, P_{FA} = 10^{-7}$.}
    \label{fig:5}
\end{figure}
\begin{figure}
    \centering
    \includegraphics[width=\columnwidth]{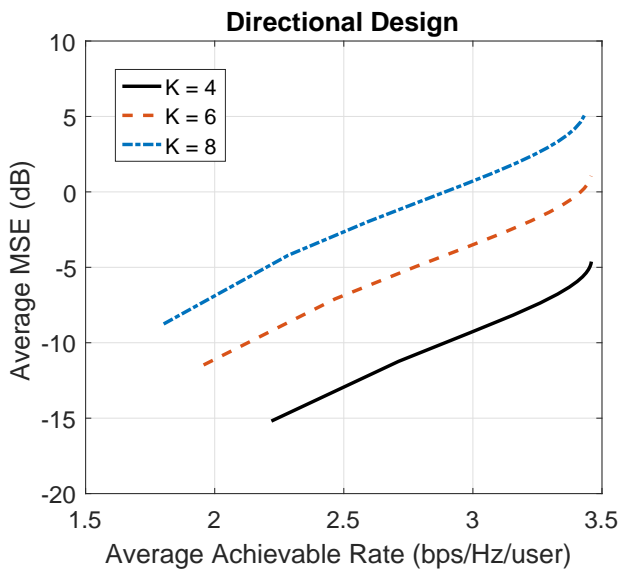}
    \caption{Trade-off of the achievable rate per user and the average MSE between the designed and desired directional beampattern, $N = 16$.}
    \label{fig:6}
\end{figure}
\\\indent In Fig. 5 and 6, we aim to explicitly show the trade-offs between the communication and radar performance. For omnidirectional beampattern design, the detection probability $P_D$ is used as the metric, where we consider the constant false-alarm rate (CFAR) detection for a point-like target in the far field, located at the angle of $\pi/5$. The receive SNR is fixed at -6dB. The false-alarm probability for radar is $P_{FA} = 10^{-7}$. We calculate the detection probability based on [9, eq. (69)]. It can be seen that there exists a trade-off between the communication rate and the radar detection performance. For a fixed $P_D$, the achievable rate increases with the decrease of the number of users, which suggests that the MUI energy can be further minimized by increasing the DoFs. The same trends appear in Fig. 6 for the directional beampattern, where we employ the mean squared error (MSE) between the desired and obtained directional beampatterns as the radar metric. Both figures prove that our approach can achieve a favorable trade-off between radar and communications.
\subsection{Dual-functional Constant Modulus Waveform Design}
\begin{figure}
    \centering
    \includegraphics[width=\columnwidth]{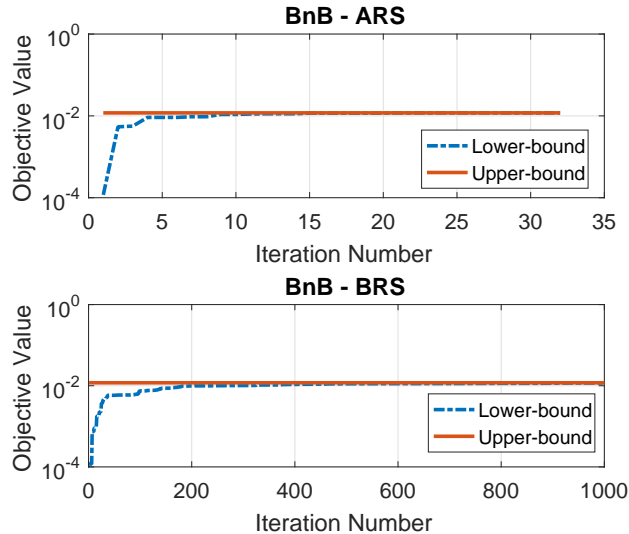}
    \caption{Convergence Behavior of BnB Algorithm for $N = 16, K = 4, \varepsilon = 1$.}
    \label{fig:7}
\end{figure}
We show the results for solving the waveform optimization problem with CMC and SC in Fig. 7-9. Following the simulation configurations in \cite{7450660}, we employ the orthogonal chirp waveform matrix as the reference signal. The convergence behavior of the proposed BnB algorithm for solving (28) is shown in Fig. 7, with $N = 16, K = 4, \varepsilon = 1$, where we compare the performance of the two different subdivision rules, i.e., ARS and BRS. Both methods converge in a finite number of iterations with a nearly constant upper-bound, which suggests that we can reach the optimal value of problem (28) by iteratively using the local algorithms for several times, e.g., QCQP solver or the proposed gradient projection method. Nevertheless, due to the non-convexity of the problem, we need BnB algorithm to confirm that this is indeed a global optimum. It can be also observed that the BnB-ARS has a faster convergence rate than BnB-BRS, which is consistent with the analysis in \cite{Tuy2016Convex}.
\begin{figure}
    \centering
    \includegraphics[width=\columnwidth]{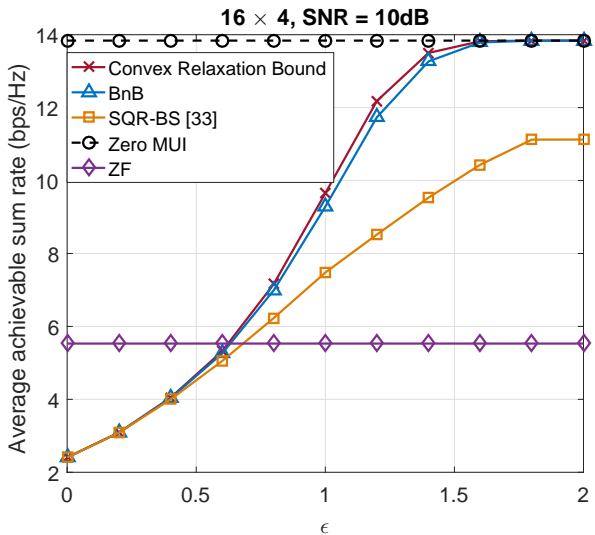}
    \caption{Trade-off between the communication sum-rate and radar waveform similarity, $N = 16, K = 4, \operatorname{SNR} = 10\text{dB}$.}
    \label{fig:9}
\end{figure}
\begin{figure}
\centering
\subfloat[]{\includegraphics[width=\columnwidth]{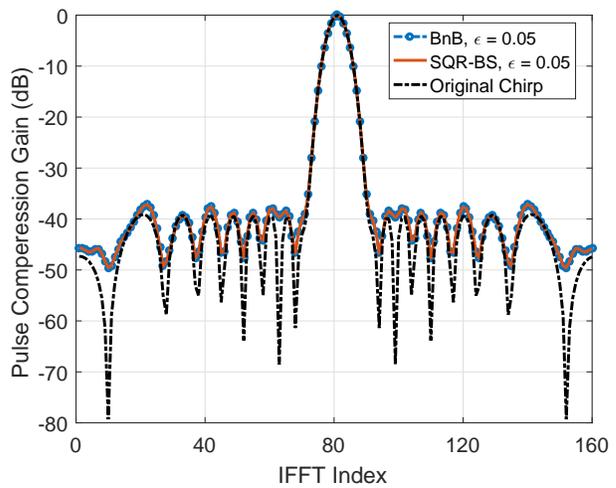}
\label{fig.10}}
\vspace{0.1in}
\hspace{.1in}
\subfloat[]{\includegraphics[width=\columnwidth]{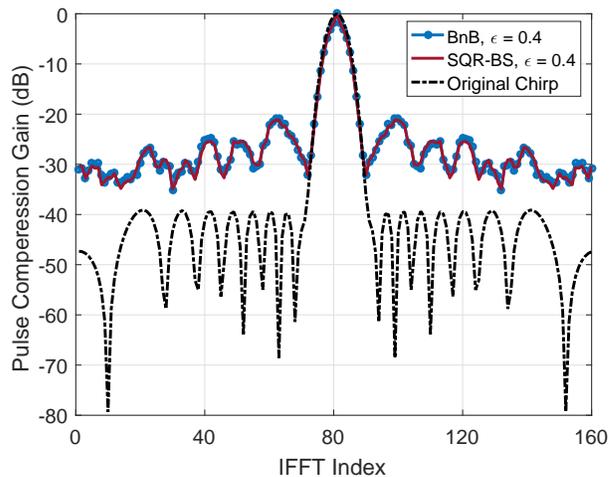}
\label{fig.11}}
\vspace{0.1in}
\hspace{.1in}
\subfloat[]{\includegraphics[width=\columnwidth]{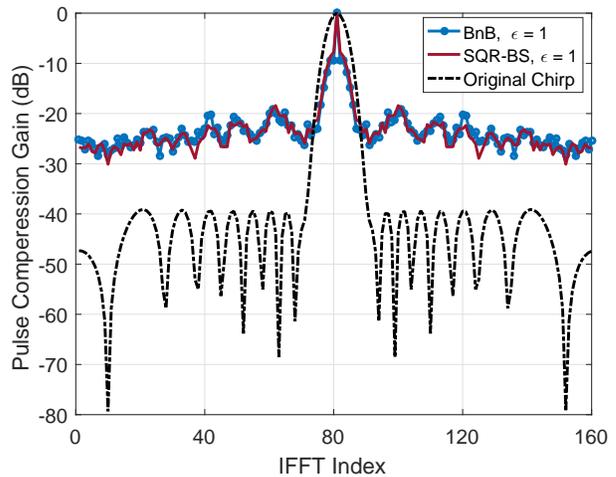}
\label{fig.12}}
\caption{Radar pulse compression for different similarity tolerance, $N = 16, K = 4$. (a) $\varepsilon = 0.05$;  (b) $\varepsilon = 0.4$; (c) $\varepsilon = 1$.}
\label{fig_radar_pc}
\end{figure}
\\\indent In Fig. 8 and 9, we show the trade-off between communication sum-rate and radar waveform similarity for the constant modulus designs, where we employ the SQR-Binary Search (SQR-BS) algorithm proposed by \cite{7450660} as our benchmark technique. Fig. 8 demonstrates the communication sum-rate with increasing $\varepsilon$ for $N = 16, K = 4, \operatorname{SNR} = 10\text{dB}$. As expected, the proposed BnB algorithm outperforms the SQR-BS significantly, since the result obtained by BnB is the global optimum, while SQR-BS can only yield local minimum solutions. It is worth highlighting that the performance of BnB is very close to the convex relaxation bound, which is obtained by solving QP-LB. When the similarity tolerance $\varepsilon$ is big enough, our BnB algorithm can approximate the AWGN capacity, i.e., the MUI can be fully eliminated.
\\\indent Fig. 9 shows the results of radar pulse compression with different similarity tolerance $\varepsilon$, where we use the waveform transmitted by the first antenna, and employ the classic FFT-IFFT pulse compression method \cite{richards2010principles} with a Taylor window to reduce the power of sidelobes. From Fig. 9 we see that there exists a trade-off between the communication sum-rate and radar pulse compression performance. Moreover, the pulse compression results of BnB and SQR-BS are nearly the same, as their performance is guaranteed by the same waveform similarity constraint, which again proves the superiority of the proposed BnB Algorithm.
\section{Conclusion}
In this paper, we discuss the waveform design for dual-functional radar-communication system, which can be used for both target detection and downlink communications. First of all, two design approaches are proposed to minimize the multi-user interference while formulating an appropriate radar beampattern, which have been further extended as a weighted optimization to achieve a flexible trade-off between radar and communications. It has been proven that the computational costs for the above three approaches are all at the same level with the communication-only ZF precoding. Numerical results show that all the proposed methods outperform the ZF precoding, while guaranteeing both the radar and communication performance. Moreover, our trade-off design can significantly improve the communication performance by allowing a slight performance loss at radar. Finally, we consider the non-convex constant modulus waveform design with similarity constraints, where an efficient global optimization algorithm based on the branch-and-bound framework has been developed. Gradient projection algorithms are used to efficiently obtain the upper and lower bounds. Simulations show that the proposed BnB algorithm for constant modulus waveform design with similarity constraints considerably outperforms the conventional SQR-BS algorithm by obtaining the global optimum of the problem.


%

\appendix[Derivation of the Projector ${\operatorname{PR}}_2$]
\begin{figure}
\centering
\subfloat[]{\includegraphics[width=0.9\columnwidth]{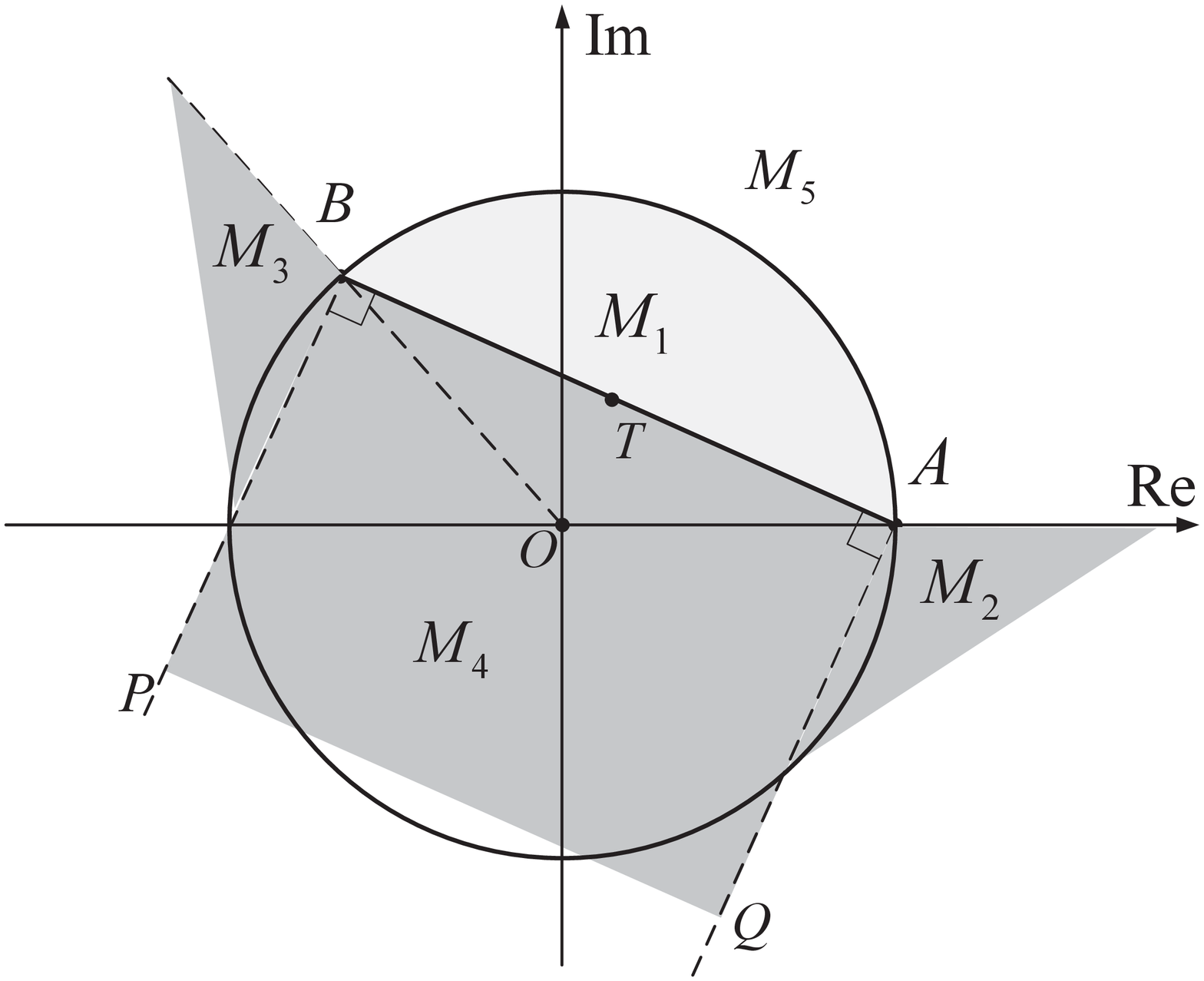}
\label{fig.4}}
\vspace{0.1in}
\hspace{.1in}
\subfloat[]{\includegraphics[width=0.85\columnwidth]{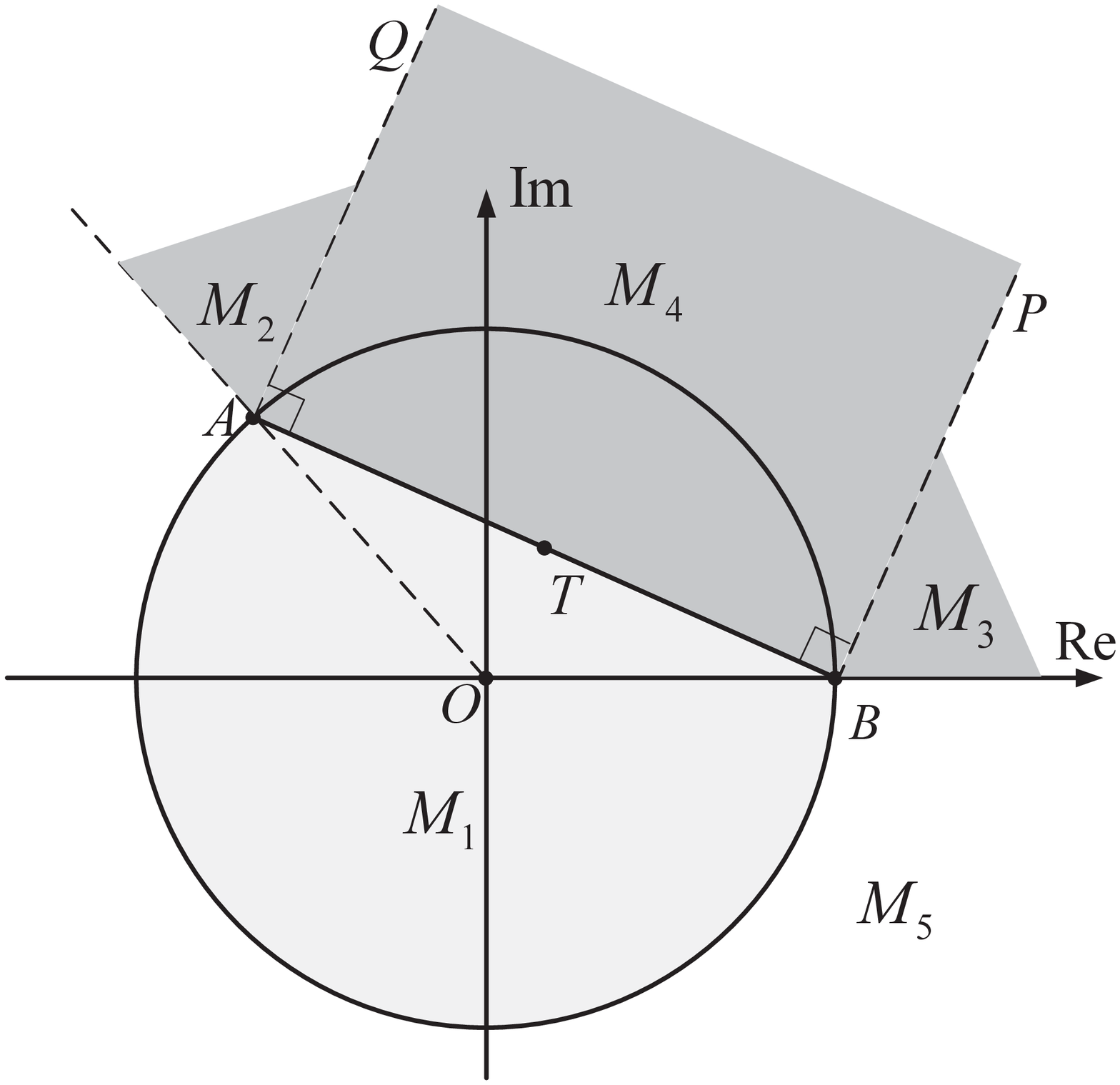}
\label{fig.5}}
\caption{Projector for GP. (a) $\phi \le \pi$; (b) $\phi \ge \pi$.}
\label{fig_projector}
\end{figure}
The projector are derived for two cases respectively, i.e., the open angle of the circular segment is (a) less than $\pi$ or (b) greater than $\pi$. We start from the first case. As shown in Fig. 10 (a), the whole complex plane $\mathbb{C}$ has been divided into five parts. The lower and the upper bounds for the angle are $l$ and $u$ respectively. Let us define
\begin{equation}
    A = \exp\left(jl\right), B = \exp\left(ju\right), T = \left(A+B\right)/2,
\end{equation}
where $T$  is the midpoint of $AB$. Given $X\in\mathbb{C}$, we aim to find a nearest point ${\operatorname{PR}}_2\left(X\right) \in M_1$ as the projection. Note that $\forall X \in M_1$, the projection is itself. For $X\in M_2$ and $X\in M_3$, the nearest points are $A$ and $B$ respectively. For $X \in M_4$, we project it onto the line $AB$, and the projection is the foot of perpendicular. For $\forall X \in {M_5} = \mathbb{C}\backslash \bigcup\limits_{i = 1}^4 {{M_i}}$, we use the normalization as its projection. By basic plane analytic geometry, we define the following lines
\begin{equation}
\begin{gathered}
  {\text{Line}}\;AB:\;\;{f_1}\left( X \right) = \operatorname{Re} \left( {{T^*}\left( {X - T} \right)} \right) = 0, \hfill \\
  {\text{Line}}\;OA:\;\;{f_2}\left( X \right) =  - \operatorname{Re} \left( {j{A^*}X} \right) = 0, \hfill \\
  {\text{Line}}\;OB:\;\;{f_3}\left( X \right) = \operatorname{Re} \left( {j{B^*}X} \right) = 0, \hfill \\
  {\text{Line}}\;AQ:\;\;{f_4}\left( X \right) = \operatorname{Re} \left( {{{\left( {A - B} \right)}^*}\left( {X - A} \right)} \right) = 0, \hfill \\
  {\text{Line}}\;BP:\;\;{f_5}\left( X \right) = \operatorname{Re} \left( {{{\left( {B - A} \right)}^*}\left( {X - B} \right)} \right) = 0. \hfill \\
\end{gathered}
\end{equation}
The projector is then given as
\begin{equation}
{\operatorname{PR}_2}\left( X \right) = \left\{ \begin{gathered}
  X,\;{f_1}\left( X \right) \ge  0,\;\left| X \right| \le  1\;\left( {X \in {M_1}} \right),\;\; \hfill \\
  A,\;\;{f_2}\left( X \right) \le  0 \le  {f_4}\left( X \right)\;\left( {X \in {M_2}} \right),\; \hfill \\
  B,\;\;{f_3}\left( X \right) \le  0 \le  {f_4}\left( X \right)\left( {X \in {M_3}} \right),\; \hfill \\
  {X_T},{f_1}\left( X \right),{f_4}\left( X \right),{f_5}\left( X \right) \le  0\;\left( {X \in {M_4}} \right),\; \hfill \\
  X/\left| X \right|,\;\text{else}, \hfill \\
\end{gathered}  \right.
\end{equation}
where $X_T$ is the foot of perpendicular on $AB$, i.e., $X{X_T}\bot AB, {X_T}\in AB $. This is given by
\begin{equation}
    X_T = X - \operatorname{Re} \left( {{{\left( {X - T} \right)}^*}T} \right)\frac{T}{{\left| T \right|}}.
\end{equation}
For the case of $\phi \ge \pi$ the projector is the same. The only difference is that $f_1\left(X\right)$ should be defined as
\begin{equation}
    {f_1}\left( X \right) = -\operatorname{Re} \left( {{T^*}\left( {X - T} \right)} \right).
\end{equation}




\ifCLASSOPTIONcaptionsoff
  \newpage
\fi



\bibliographystyle{IEEEtran}
\bibliography{IEEEabrv,CEP_REF}
\end{document}